\documentclass[submission,copyright,creativecommons]{eptcs}
\providecommand{\event}{TARK 2023} 

\usepackage{iftex}

\ifpdf
  \usepackage{underscore}         
  \usepackage[T1]{fontenc}        
\else
  \usepackage{breakurl}           
\fi

\title{Resilient Information Aggregation}
\author{
Itai Arieli
\institute{Technion \\ Haifa, Israel}
\email{iarieli@technion.ac.il}
\and
Ivan Geffner
\institute{Technion \\ Haifa, Israel}
\email{ieg8@cornell.edu}
\and
Moshe Tennenholtz
\thanks{The work by Ivan Geffner and Moshe Tennenholtz was supported by funding from
the European Research Council (ERC) under the European Union’s Horizon 2020
research and innovation programme (grant agreement 740435).}
\institute{Technion \\ Haifa, Israel}
\email{moshet@technion.ac.il}
}
\def\titlerunning{Resilient Information Aggregation}
\def\authorrunning{I. Arieli, I. Geffner \& M. Tennenholtz}

\usepackage{amsmath}
\usepackage{amsfonts}
\usepackage{graphicx}
\usepackage{float}
\usepackage[english]{babel}
\usepackage{amsthm}
\newtheorem{theorem}{Theorem}
\newtheorem{proposition}{Proposition}
\newtheorem{definition}{Definition}
\newtheorem{corollary}{Corollary}
\newtheorem{lemma}{Lemma}
\newtheorem{example}{Example}
\newcommand{\commentout}[1]{}

\begin{document}
\maketitle

\begin{abstract}
In an information aggregation game, a set of senders interact with a receiver through a mediator. Each sender observes the state of the world and communicates a message to the mediator, who recommends an action to the receiver based on the messages received. The payoff of the senders and of the receiver depend on both the state of the world and the action selected by the receiver. This setting extends the celebrated cheap talk model in two aspects: there are many senders (as opposed to just one) and there is a mediator. From a practical perspective, this setting captures platforms in which strategic experts advice is aggregated in service of action recommendations to the user. We aim at finding an optimal mediator/platform that maximizes the users' welfare given highly resilient incentive compatibility requirements on the equilibrium selected: we want the platform to be incentive compatible for the receiver/user when selecting the recommended action, and we want it to be resilient against group deviations by the senders/experts. We provide highly positive answers to this challenge, manifested through efficient algorithms.
\end{abstract}

\section{Introduction}

Experts' opinions aggregation platforms are crucial for web monetizing. 
Indeed, in sites such as Reddit or Google, comments and reviews are aggregated as an answer to a user query about items observed or studied by others. We refer to these reviewers as \emph{experts}.
The platform can aggregate these experts' inputs or filter them when providing a recommendation to the user, which will later lead to a user action. An ideal platform should maximize the users' social welfare. In an economic setting, however, the different experts may have their own preferences. Needless to say, when commenting on a product or a service, we might not know if the expert prefers the user to buy the product or accept the service, or if the expert prefers otherwise. This is true even when all experts observe exactly the same characteristics of a product or service.  

Interestingly, while the study of economic platforms is rich \cite{Papadimitriou2022,khan2021disinformation,singh2021fulfillment,aouad2020online,olivares2020balancing,stouras2020prizes,derakhshan2022product,li2022interference,shi2021recommender}, there is no rigorous foundational and algorithmic setting for the study of aggregation and filtering of strategic experts opinions in service of the platform users. In this paper, we initiate such a study, which we believe to be essential. This study can be viewed as complementary to work on platform incentives \cite{Papadimitriou2022}, issues of dishonesty \cite{khan2021disinformation}, and issues of ranking/filtering \cite{derakhshan2022product}, by putting these ingredients in a concrete foundational economic setting dealing with recommendations based on 
inputs from strategic experts.
The model we offer extends the classical cheap talk model in two fundamental directions. First, by having several strategic senders (experts) rather than only one; second, by introducing a platform that acts as a mediator in an information design setting.

Our work is related to the literature on information design that studies optimal information disclosure policies for informed players. The two leading models of information design are cheap talk \cite{crawford1982strategic} and Bayesian persuasion \cite{kamenica2011bayesian}. The main distinction between
these models is the underlying assumption that, in the Bayesian persuasion model, the sender has commitment power in the way she discloses the information, while in the cheap talk model she has not.

Bayesian persuasion models emphasize commitment power, and while it may hold in some real-world situations, it is often considered strong. In addition, in Bayesian persuasion, the informed agent (the sender) is also the one who designs the information revelation policy. In practice, however, information revelation can be determined by other external or legal constraints. A commerce platform, for example, determines what information about a product is revealed to a potential customer based on information submitted by different suppliers.

In our model there is a finite state space of size $n$, several informed players (senders), an uninformed player (the receiver) that determines the outcome of the game by playing a binary action from the set $A := \{0,1\}$ (this could represent buying a product or not, passing a law or not, etc.), and a mediator that acts as a communication device between the senders and the receiver (the mediator can be seen as the platform used by all parties). The utility of each player is determined by the state and by the action played by the receiver. The incentives of the senders may not necessarily be aligned (e.g., senders can be a car seller and a technician that tested the car, two independent parties who studied the monetary value of law, two suppliers of a product, etc.).
The state is drawn from a prior distribution that is commonly known among players, but only the senders know its realized value. Thus, the senders' purpose is to reveal information to the receiver in such a way that the receiver plays the action that benefits them the most. Since the senders have no commitment power, we are interested in a mediated cheap talk equilibrium, in which it is never in the interest of the senders to be dishonest, and it is always in the interest of the receiver to play the action suggested by the protocol.

The most common notions of equilibrium, such as Nash equilibrium, require that each individual player cannot increase its utility by deviating from the proposed strategy. However, notions of equilibria that are resilient to group deviations are currently gaining traction \cite{aiyer2005bar, halpern2008beyond, aah2011},
in particular because of their Web applications. Indeed, on the Internet, it is not only fairly easy to collude, but it is also relatively simple to create proxy pseudo-identities and defect in a coordinated way (this is known as a Sybil attack~\cite{Douceur2002}). Nowadays, in Web applications and in distributed systems, resilience against individual deviations is generally considered insufficient for practical purposes. For instance, blockchain protocols are required to tolerate coordinated deviations from up to a fraction of their user base. In this work, we focus on $k$-resilient equilibria, which are strategies profiles in which no coalition of up to $k$ players can increase their utility by deviating.

Our main goal in the paper is to characterize, given the incentives of the senders and the receiver, which maps from states to distributions over actions result from playing $k$-resilient equilibria. More precisely, each cheap talk protocol $\vec{\sigma}$ induces a map $M$ from states to distributions over actions, where $M(\omega)$ is mapped to the distribution over actions resulting from playing $\vec{\sigma}$ in state $\omega$. Our aim is to characterize which of these maps (or \emph{outcomes}, as we call them) can be implemented by a $k$-resilient equilibrium, and to efficiently construct a concrete $k$-resilient equilibrium whenever a given outcome is implementable. We first show that, if there are more than two senders, even if one of them defects and misreports the state, a majority of the senders would still report the truth, and thus the mediator will always be able to compute the correct state. Therefore, if there are at least three senders, outcomes are implementable by a $1$-resilient equilibrium (i.e., a Nash equilibrium) if and only if they are incentive-compatible for the receiver. That is, an outcome is implementable  by a $1$-resilient equilibrium if and only if it improves the utility of the receiver relative to the case where no information is revealed to her. This result implies that the set of implementable distributions is independent of the utilities of the senders and only depends on that of the receiver, and thus that the senders have no bargaining power. It is also easy to check that this result  extends to the case of $k$-resilient equilibria for $k < n/2$, where $n$ is the number of senders. However, we show that if a majority of the players can collude, the set of implementable outcomes is defined by a system of linear equations that depend both on the utilities of the senders and the receiver. It may seem at first that computing such characterization may be highly inefficient since the number of possible coalitions of size at most $k \ge n/2$ grows exponentially over the number of players, and each of these possible coalitions imposes a constraint on the outcome. By contrast, our main result shows that, if the number of states is $m$, then the aforementioned linear system can be written with only $m^2$ inequality constraints, and all such inequalities can be computed in polynomial time over $m$ and the number of senders $n$. This means that the best receiver $k$-resilient equilibrium or the $k$-resilient equilibrium that maximizes social welfare can be computed efficiently. We also provide, given a solution of the system of equations, an efficient way to construct a concrete $k$-resilient equilibrium that implements the desired outcome.

Our results so far assume that all senders have full information about the realized state. However, in some cases it is realistic to assume that senders only have partial information about it and, moreover, that each sender's information might be different. We show in Section~\ref{sec:extended} that our techniques generalize to this model as long as the senders's preferences are not influenced by their coalition, a condition that we call \emph{$k$-separability}. This means that, assuming $k$-separability, we provide a characterization of all outcomes that are implementable by a $k$-resilient equilibrium, and an algorithm that construct a concrete $k$-resilient equilibrium that implements a desired (implementable) outcome. Both the characterization and the algorithm are efficient relative to the size of the game's description. 

\subsection{Related Work}

The literature on information design is too vast to address all the related work. We will therefore mention some key related papers. Krishna and Morgan \cite{krishna2001model} consider a setting similar to that considered by Crawford and Sobel \cite{crawford1982strategic}, where a real interval represents the set of states and actions. In this setting, the receiver's and the senders' utilities are \emph{biased} by some factor that affects their incentives and utility. Similarly to the current paper where the sender is not unique, Krishna and Morgan consider two informed senders that reveal information sequentially to the receiver. They consider the best receiver equilibrium and show that, when both senders are \emph{biased} in the same direction, it is never beneficial to consult both of them. By contrast, when senders are biased in opposite directions, it is always beneficial to consult them both.
 
In another work, Salamanca \cite{salamanca2021value} characterizes the optimal mediation for the sender in a sender-receiver game. Lipnowski and Ravid \cite{lipnowski2020cheap}, and 
 Kamenica and Gentzkow \cite{kamenica2011bayesian} provide a geometric characterization of the best cheap talk equilibrium for the sender under the assumption that the sender's utility is state-independent. The geometric characterization of Lipnowski and Ravid is no longer valid for the case where there are two or more senders. 

Kamenika and Gentzkow \cite{kamenica2017competition} consider a setting with two senders in a Bayesian persuasion model. The two senders, as in the standard Bayesian persuasion model (and unlike ours), have commitment power and they compete over information revelation.
The authors characterize the equilibrium outcomes in this setting.

In many game-theoretical works, mediators are incorporated into strategic settings \cite{aumann1987correlated,morgan1999models}. 
Kosenko \cite{kosenko2018mediated} also studied the information aggregation problem. However, their model assumed that the mediator had incentives of its own and selected its policy at the same time as the sender. Monderer and Tennenholtz \cite{monderer2009strong} studied the use of mediators to enhance the set of situations where coalition deviance is stable. They show that using mediators in several classes of settings can produce stable behaviors that are resistant to coalition deviations. 
In our setting, the existence of a $k$-resilient equilibrium is straightforward (e.g., playing a constant action). Instead, the strength of our result follows from efficiently characterising the set of all outcomes that are implementable using $k$-resilient mediated equilibria.

\section{Model}\label{sec:model}

In an information aggregation game $\Gamma = (S, A, \Omega, p, u)$, there is a finite set of possible states $\Omega = \{\omega^1, \ldots, \omega^m\}$, a commonly known distribution $p$ over $\Omega$, a set of possible actions $A = \{0,1\}$, a set $S = \{1,2, \ldots, n\}$ of senders, a receiver $r$, a mediator $d$, and a utility function $u : (S \cup \{r\}) \times \Omega \times A \longrightarrow \mathbb{R}$ such that $u(i, \omega, a)$ gives the utility of player $i$ when action $a$ is played at state $\omega$. Each information aggregation game instance is divided into four phases. In the first phase, a state $\omega$ is sampled from $\Omega$ following distribution $p$ and this state is disclosed to all senders $i \in S$. During the second phase, each sender $i$ sends a message $m_i$ to the mediator. In the third phase (after receiving a message from each sender) the mediator must send a message $m_d \in A$ to the receiver,  and in the last phase the receiver must play an action $a \in A$ and each player $i \in S \cup \{r\}$ receives $u(i, \omega, a)$ utility. 

The behavior of each player $i$ and is determined by its strategy $\sigma_i$ and the behavior of the mediator is determined by its strategy $\sigma_d$. A strategy $\sigma_i$ for a player $i \in S$ can be represented by a (possibly randomized) function $m_i: \Omega \longrightarrow \{0,1\}^*$ such that $m_i(\omega)$ indicates what message $i$ is sending to the mediator given state $\omega \in \Omega$. The strategy $\sigma_d$ of the mediator can be represented by a function $m_d: \left(\{0,1\}^*\right)^n \longrightarrow A$ that indicates, given the message received from each player, what message it should send to the receiver. The strategy $\sigma_r$ of the receiver can be represented by a function $a_r: A \rightarrow A$ that indicates which action it should play given the message received from the mediator.

In summary, a game instance goes as follows:
\begin{enumerate}
    \item A state $\omega$ is sampled from $\Omega$ following distribution $p$, and $\omega$ is disclosed to all senders $i \in S$.
    \item Each sender $i \in S$ sends message $m_i(\omega)$ to the mediator.
    \item The mediator sends message $m_d(m_1, \ldots, m_n)$ to the receiver.
    \item The receiver plays action $a_r(m_d)$ and each player $i \in S \cup \{r\}$ receives $u(i, \omega, a_r(m_d))$ utility.
\end{enumerate}

Note that, in order to simplify the notation, we use a slight notation overload since $m_i$ is both the message sent by player $i$ and a function that depends on the state. This is because the message sent by $i$ always depend on the state, even if it is not explicitly written. A similar situation happens with $a_r$.

\subsection{Game mechanisms}

Given a game $\Gamma = (S, A, \Omega, p, u)$, a \emph{mechanism} $M = (m_1, m_2, \ldots,\allowbreak m_n, m_d, a_r)$ uniquely determines a map $o_M^\Gamma : \Omega \longrightarrow \Delta A$ (where $\Delta A$ is the set of probability distributions of $A$) that maps each state $\omega$ to the distribution of actions obtained by playing $\Gamma$ when the senders, the mediator and the receiver play the strategies represented by the components of $M$. We say that $M$ \emph{implements} $o_M^\Gamma$ and that $o_M^\Gamma$ is the \emph{outcome} of $M$.

A mechanism $M$ is \emph{incentive-compatible} if it is not in the interest of the receiver or any of the senders to deviate from the proposed mechanism (note that the mediator has no incentives). We also say that $M$ is \emph{honest} if (a) $m_i \equiv Id_{\Omega}$, where $Id_{\Omega}(\omega) = \omega$ for all $\omega \in \Omega$, and (b) $a_r \equiv Id_A$. Moreover, we say that $M$ is \emph{truthful} if it is both honest and incentive-compatible. Intuitively, a mechanism is truthful if sending the true state to the mediator is a dominant strategy for the senders and playing the state suggested by the mediator is a dominant strategy for the receiver.

\begin{example}\label{ex:basic-ex}
Consider a game $\Gamma = (S, A, \Omega, p, u)$ where $S = \{1,2,3\}$, $A = \{0,1\}$, $\Omega = \{\omega_1, \ldots, \omega_k\}$, $p$ is the uniform distribution over $\Omega$ and $u : (S \cup \{r\}) \times \Omega \times A \longrightarrow \mathbb{R}$ is an arbitrary utility function. Consider the truthful mechanism in which senders disclose the true state to the mediator, the mediator chooses the state $\omega \in \Omega$ sent by the majority of the senders and sends to the receiver the action $a$ that maximizes $u(r, \omega, a)$, and the receiver plays the action sent by the mediator. It is easy to check that this mechanism is incentive-compatible: no individual sender can influence the outcome by deviating since the mediator chooses the state sent by the majority of the senders. Moreover, by construction, this mechanism gives the receiver the maximum possible utility among all mechanisms.  
\end{example}

Our first goal is to characterize the set of possible outcomes that can be implemented by truthful mechanisms. Note that, because of Myerson's revelation principle~\cite{myerson79}, characterizing the set of outcomes implemented by truthful mechanisms is the same as characterizing the set of outcomes implemented by any incentive-compatible mechanisms (not necessarily truthful):

\begin{proposition}
Let $\Gamma = (S, A, \Omega, p, u)$ be an information aggregation game. Then, for any incentive-compatible mechanism $M$ for $\Gamma$ there exists a truthful mechanism $M'$ such that $M'$ implements $o_M^\Gamma$.
\end{proposition}

\begin{proof}
Given $M = (m_1, m_2, \ldots,\allowbreak m_n, m_d, a)$, consider a mechanism $M' = (m_1', m_2', \ldots,\allowbreak m_n', m_d', a')$ such that $m_i' \equiv Id_\Omega$ for all $i \in S$, $m_a' \equiv Id_A$, and the mediator does the following. After receiving a message $\omega_j$ from each sender $j$, it computes $a(m_d(m_1(\omega_1),\allowbreak, m_2(\omega_2), \ldots, m_n(\omega_n)))$ and sends this action to the receiver (if the message from some player $j$ is inconsistent, the mediator takes $\omega_j$ to be an arbitrary element of $\Omega$). By construction, $M'$ is a truthful mechanism in which the mediator simulates everything the players would have sent or played with $M$. It is easy to check that, with $M'$, for any possible deviation for player $j \in S \cup \{r\}$, there exists a deviation for $j$ in $M$ that produces the same outcome. Thus, if $M$ is incentive-compatible, so is $M'$.
\end{proof}

This proposition shows that we can restrict our search to truthful mechanisms. Moreover, the construction used in the proof shows that we can assume without loss of generality that the senders can only send messages in $\Omega$ since sending any other message is equivalent to sending an arbitrary element of $\Omega$. To simplify future constructions, we'll use this assumption from now on.

\subsection{Resilient equilibria}

Traditionally, in the game theory and mechanism design literature, the focus has always been on devising strategies or mechanisms such that no individual agent is incentivized to deviate. However, in the context of multi-agent Bayesian persuasion, this approach is not very interesting. The reason is that, if $n > 2$, the mediator can always compute the true state by taking the one sent by a majority of the senders (as seen in Example~\ref{ex:basic-ex}), and thus the mediator can make a suggestion to the receiver as a function of the true state while individual senders cannot influence the outcome by deviating. In fact, given action $a \in A$, let $U_a := \mathbb{E}_{\omega \leftarrow p}[u(r, \omega, a)]$ be the expected utility of the receiver when playing action $a$ regardless of the mediator's suggestion and, given outcome $o^\Gamma$, let $$E_i(o^\Gamma) := \mathbb{E}_{\substack{\omega \leftarrow p, \\ a \leftarrow o^\Gamma(\omega)}}[u(i, \omega, a)]$$ be the expected utility of player $i \in S \cup \{r\}$ with outcome $o^\Gamma$. The following proposition characterizes outcomes implementable by truthful mechanisms.

\begin{proposition}\label{prop:no-coalitions}
If $\Gamma = (S, A, \Omega, p, u)$ is an information aggregation game with $|S| > 2$, an outcome $o^\Gamma: \Omega \longrightarrow \Delta A$ of $\Gamma$ is implementable by a truthful mechanism if and only if $E_r\left(o^\Gamma\right) \ge U_a$ for all $a \in A$.
\end{proposition}

Intuitively, proposition~\ref{prop:no-coalitions} states that, if there are at least three senders, the only condition for an outcome to be implementable by a truthful incentive-compatible mechanism is that the receiver gets a better expected utility than the one it gets with no information.  Before proving it, we need the following lemma, which will also be useful for later results.

\begin{lemma}\label{ref:lemma-receiver-compatibility}
Let $\Gamma = (S, A, \Omega, p, u)$ be an information aggregation game. An honest mechanism $M = (Id_\Omega, \ldots, Id_\Omega, m_d, Id_A)$ for $\Gamma$ is incentive-compatible for the receiver if and only if $E_r\left(o_M^\Gamma\right) \ge U_a$ for all $a \in A$.
\end{lemma}

\begin{proof}
$(\Longrightarrow)$ Let $M$ be an honest mechanism for $\Gamma$ that is incentive-compatible for the receiver. Then, if $E_r\left(o_M^\Gamma\right) < U_a$ for some $a \in A$, the receiver can increase its utility ignoring the mediator's suggestion and playing always action $a$. This would contradict the fact that $M$ is incentive-compatible.

$(\Longleftarrow)$ Suppose that $E_r\left(o_M^\Gamma\right) \ge U_a$ for all $a \in A$. If $M$ is not incentive-compatible, it means that the receiver can strictly increase its payoff either (a) by playing $1$ when the mediator sends $0$ and/or (b) playing $0$ when the mediator sends $1$. Suppose that (a) is true, then the receiver can strictly increase its payoff by playing $1$ in all scenarios, which would contradict the fact that its expected payoff with $M$ is greater or equal than $U_1$. The argument for (b) is analogous.
\end{proof}

With this we can prove Proposition~\ref{prop:no-coalitions}. The mechanism used in the proof is very similar to the one in Example~\ref{ex:basic-ex}.

\begin{proof}[Proof of Proposition~\ref{prop:no-coalitions}]
Let $M$ be a truthful mechanism. Then, by Lemma~\ref{ref:lemma-receiver-compatibility}, $o_m^\Gamma$ satisfies that $E_r\left(o_M^\Gamma\right) \ge U_a$ for all $a \in A$.

Conversely, suppose that an outcome $o^\Gamma$ satisfies that $E_r\left(o_M^\Gamma\right) \ge U_a$ for all $a \in A$. Consider a mechanism $M = (Id_\Omega, \ldots, Id_\Omega, m_d, Id_A)$ such that the mediator takes the state $\omega$ sent by the majority of the senders and sends $o^\Gamma(\omega)$ to the receiver. By construction, $M$ implements $o^\Gamma$. Moreover, as in Example~\ref{ex:basic-ex}, $M$ is incentive-compatible for the senders since, if $n > 2$, they cannot influence the outcome by individual deviations. By Lemma~\ref{ref:lemma-receiver-compatibility} $M$ is also incentive-compatible for the receiver. Thus, $M$ is a truthful mechanism that implements $o^\Gamma$.
\end{proof}

The construction used in the proof shows how easily we can implement any desired outcome as long as it is better for the sender than playing a constant action. However, Proposition~\ref{prop:no-coalitions} is only valid under the assumption that senders cannot collude and deviate in a coordinated way (an assumption that many times is unrealistic, as pointed out in the introduction). If we remove this assumption, the \emph{next best thing} is to devise mechanisms such that all coalitions up to a certain size do not get additional utility by deviating. We focus mainly on the following notions of equilibrium:

\begin{definition}[\cite{ADGH06}]
Let $\Gamma$ be any type of game for $n$ players with strategy space $A = A_1 \times \ldots \times A_n$ and functions $u_i: S \longrightarrow \mathbb{R}$ that give the expected utility of player $i$ when players play a given strategy profile. Then,
\begin{itemize}
    \item A strategy profile $\vec{\sigma} \in A$ is a \emph{$k$-resilient Nash equilibrium} if, for all coalitions $K$ up to $k$ players and all strategy profiles $\vec{\tau}_K$ for players in $K$, $u_i(\vec{\sigma}) \ge u_i(\vec{\sigma}_{-K}, \vec{\tau}_K)$ for some $i \in K$.
    \item A strategy profile $\vec{\sigma} \in A$ is a \emph{strong $k$-resilient Nash equilibrium} if, for all coalitions $K$ up to $k$ players and all strategy profiles $\vec{\tau}_K$ for players in $K$, $u_i(\vec{\sigma}) \ge u_i(\vec{\sigma}_{-K}, \vec{\tau}_K)$ for all $i \in K$.
\end{itemize}
\end{definition}

Intuitively, a strategy profile is $k$-resilient if no coalition of up to $k$ players can deviate in such a way that all members of the coalition strictly increase their utility, and a strategy profile is strongly $k$-resilient if no member of any coalition of up to $k$ players can strictly increase its utility by deviating, even at the expense of the utility of other members of the coalition. We can construct analogous definitions in the context of information aggregation:

\begin{definition}
Let $\Gamma = (S, A, \Omega, p, u)$ be an information aggregation game. A mechanism $M = (m_1, \ldots, \allowbreak m_n,\allowbreak m_d, a_r)$ for $\Gamma$ is $k$-resilient incentive-compatible (resp., strong $k$-resilient incentive-compatible) if 
\begin{itemize}
\item[(a)] The receiver cannot increase its utility by deviating from the proposed protocol.
\item[(b)] Fixing $m_d$ and $a_r$ beforehand, the strategy profile of the senders determined by $M$ is a $k$-resilient Nash equilibrium (resp., strong $k$-resilient Nash equilibrium).
\end{itemize}
\end{definition}

A mechanism $M$ is $k$-resilient truthful if it is honest and $k$-resilient incentive-compatible. Strong $k$-resilient truthfulness is defined analogously.

\section{Main Results}\label{sec:results}

For the main results of this paper we need the following notation. Given an outcome $o: \Omega \rightarrow \Delta A$, we define by $o^*: \Omega \rightarrow [0,1]$ the function that maps each state $\omega$ to the probability that $o(\omega) = 0$. Note that, since $|A| = 2$, $o$ is uniquely determined by $o^*$. The following theorem gives a high level characterization of all $k$-resilient truthful mechanisms (resp., strong $k$-resilient truthful mechanisms).

\begin{theorem}\label{thm:high-level-char}
Let $\Gamma = (S, A, \Omega, p, u)$ be an information aggregation game with $\Omega = \{\omega^1, \ldots, \omega^m\}$. Then, there exists a system $E$ of $O(m^2)$ equations over variables $x_1, \ldots, x_m$, such that each equation of $E$ is of the form $x_i \le x_j$ for some $i,j \in [m]$, and such that an outcome $o$ of $\Gamma$ is implementable by a $k$-resilient truthful mechanism (resp., strong $k$-resilient truthful mechanism) if and only if 

\begin{itemize}
\item[(a)] $x_1 = o^*(\omega^1), \ldots, x_m = o^*(\omega^m)$ is a solution of $E$.
\item[(b)] $E_r\left(o\right) \ge U_a$ for all $a \in A$.
\end{itemize}

Moreover, the equations of $E$ can be computed in polynomial time over $m$ and the number of senders $n$.
\end{theorem}

Note that condition (b) is identical to the one that appears in Lemma~\ref{ref:lemma-receiver-compatibility}. In fact, condition (b) is the necessary and sufficient condition for a mechanism that implements $o$ to be incentive-compatible for the receiver, and condition (a) is the necessary and sufficient condition for this mechanism to be $k$-resilient incentive-compatible (resp., strong $k$-resilient incentive-compatible) for the senders.
Theorem~\ref{thm:high-level-char} shows that the set of outcomes implementable by $k$-resilient truthful mechanisms (resp., strong $k$-resilient truthful mechanisms) is precisely the set of solutions of a system of equations over $\{o^*(\omega^i)\}_{i \in [m]}$. This means that the solution that maximizes any linear function over $\{o^*(\omega^i)\}_{i \in [m]}$ can be reduced to an instance of linear programming. In particular, the best implementable outcome for the receiver or for each of the senders can be computed efficiently.

\begin{corollary}
There exists a polynomial time algorithm that computes the outcome that could be implemented by a $k$-resilient truthful mechanism (resp., strong $k$-resilient truthful mechanism) that gives the most utility to the receiver or that gives the most utility to a particular sender.
\end{corollary}

Our last result states that not only we can characterize the outcomes implementable by truthful mechanisms, but that we can also efficiently compute a truthful mechanism that implements a particular outcome. 
Before stating this formally, it is important to note that all truthful mechanisms can be encoded by a single function $m_d^*$ from message profiles $\vec{m} = (m_1, \ldots, m_n)$ to $[0,1]$. Intuitively, the mechanism $m_d$ defined by $m_d^*$ is the one that maps $(\vec{m})$ to the distribution such that $0$ has probability $m_d^*(\vec{m})$ and $1$ has probability $1 - m_d^*(\vec{m})$. Moreover, note that the description of a $k$-resilient truthful mechanism for a game with $m$ possible states is exponential over $k$ since the mechanism must describe what to do if $k$ players misreport their state, which means that the mechanism should be defined over at least $m^k$ inputs. Clearly, no polynomial algorithm over $n$ and $m$ can compute this mechanism just because of the sheer size of the output. However, given a game $\Gamma$ and an output $o$, it is not necessary to compute the whole description of the resilient truthful mechanism $m_d^*$ for $\Gamma$ that implements $o$, we only need to be able to compute $m_d^*(\vec{m})$ in polynomial time for each possible message profile $\vec{m}$. We state this as follows.

\begin{theorem}\label{thm:construction}
There exists an algorithm $\pi$ that receives as input the description of an information aggregation game  $\Gamma = (S, A, \Omega, p, u)$, an outcome $o$ for $\Gamma$ implementable by a $k$-resilient  mechanism (resp., strong $k$-resilient mechanism), and a message input $\vec{m}$ for the mediator, and $\pi$ outputs a value $q \in [0,1]$ such that the function $m_d^*$ defined by $m_d^*(\vec{m}) := A(\Gamma, o, \vec{m})$ determines a $k$-resilient truthful mechanism (resp., strong $k$-resilient truthful mechanism) for $\Gamma$ that implements $o$. Moreover, $\pi$ runs in polynomial time over $|\Omega|$ and $|S|$.
\end{theorem}

The proofs of Theorems~\ref{thm:high-level-char} and \ref{thm:construction} are detailed in Sections~\ref{sec:proofs} and \ref{sec:proof-construction} respectively. Intuitively, each coalition imposes a constraint over the space of possible messages that the mediator may receive, implying that the mediator should suggest action $0$ more often for some message inputs than others. These constraints induce a partial order over \emph{pure inputs} (i.e., messages such that all senders report the same state), which is precisely the order defined by $E$ in Theorem~\ref{thm:high-level-char}. It can be shown that, even though there may be exponentially many possible coalitions of size at most $k$, this partial order can be computed in polynomial time over the number of states and senders.

\section{Proof of Theorem~\ref{thm:high-level-char}}\label{sec:proofs}

In this section we prove Theorem~\ref{thm:high-level-char}. Note that, because of Lemma~\ref{ref:lemma-receiver-compatibility}, we only have to show that, given a game $\Gamma = (S, A, \Omega, p, u)$ with $|\Omega| = m$ and $|S| = n$, there exists a system of equations $E$ as in Theorem~\ref{thm:high-level-char} such that an outcome $o$ is implementable by an honest mechanism that is $k$-resilient incentive-compatible (resp., strong $k$-resilient) for the senders if and only if $(o^*(\omega^1), \ldots, o^*(\omega^m))$ is a solution of $E$.

To understand the key idea, let us start with an example in which $\Omega = \{\omega^1, \omega^2\}$, $S = \{1,2,3,4\}$, senders $1,2$ and $3$ prefer action $0$ in $\omega^2$, senders $2,3$ and $4$ prefer action $1$ in $\omega^1$, and in which we are trying to characterize all outcomes that could be implemented by a mechanism that is $2$-resilient incentive-compatible for the senders. If all senders are honest, then the mediator could only receive inputs $(\omega^1,\omega^1,\omega^1,\omega^1)$ or $(\omega^2,\omega^2,\omega^2,\omega^2)$ (where the $i$th component of an input represents the message sent by sender $i$). However, since senders could in principle deviate, the mediator could receive, for instance, an input of the form $(\omega^1,\omega^1,\omega^2,\omega^2)$. This input could originate in two ways, either the true state is $\omega^1$ and senders 3 and 4 are misreporting the state, or the state is $\omega^2$ and senders $1$ and $2$ are misreporting. Even though a mechanism is honest, the mediator's message function $m_d$ should still be defined for inputs with different components, and it must actually be done in such a way that players are not incentivized to misreport.

Let $m_d^*$ be the function that maps each message $(m_1, m_2, m_3, m_4)$ to the probability that $m_d(m_1, \ldots,\allowbreak  m_4) = 0$. If the honest mechanism determined by $m_d^*$ is $2$-resilient incentive-compatible for the senders, the probability of playing $0$ should be lower with $(\omega^1,\omega^1,\omega^2,\omega^2)$ than with $(\omega^2,\omega^2,\omega^2,\omega^2)$. Otherwise, in $\omega^2$, senders $1$ and $2$ can increase their utility by reporting $1$ instead of $2$. Thus, $m_d^*$ must satisfy that $m_d^*(\omega^1,\omega^1,\omega^2,\omega^2) \le m_d^*(\omega^2,\omega^2,\omega^2,\omega^2)$. Moreover, $m_d^*(\omega^1,\omega^1,\omega^2,\omega^2) \ge m_d^*(\omega^1,\omega^1,\omega^1,\omega^1)$, since otherwise, in state $\omega^1$, senders $3$ and $4$ can increase their utility by reporting $2$ instead of $1$. These inequalities together imply that $m_d^*(\omega^1,\omega^1,\omega^1,\omega^1) \le m_2^*(\omega^2,\omega^2,\omega^2,\omega^2)$, and therefore that $o^*(\omega^1) \le o^*(\omega^2)$. In fact, we can show that this is the only requirement for $o$ to be implementable by a mechanism that is $k$-resilient incentive compatible for the senders. Given $o$ such that $o^*(\omega^1) \le o^*(\omega^2)$, consider an honest mechanism determined by $m_d^*$, in which $m_d^*(m_1, m_2, m_3, m_4)$ is defined as follows:

\begin{itemize}
    \item If at least three players sent the same message $\omega$, then $m_d^*(m_1, \allowbreak m_2, m_3, m_4) := o^*(\omega)$.
    \item Otherwise, $m_d^*(m_1, m_2, m_3, m_4) := (o^*(\omega^1) + o^*(\omega^2))/2$.
\end{itemize}

We can check that the honest mechanism $M$ determined by $m_d^*$ is indeed $2$-resilient incentive-com\-pa\-ti\-ble for the senders. Clearly, no individual sender would ever want to deviate since it cannot influence the outcome by itself (still three messages would disclose the true state). Moreover, no pair of senders can increase their utility by deviating since, in both $\omega^1$ and $\omega^2$, at least one of the senders in the coalition would get the maximum possible utility by disclosing the true state. This shows that, in this example, $o^*(\omega^1) \le o^*(\omega^2)$ is the only necessary and sufficient condition for $o$ to be implementable by a mechanism that is $2$-resilient incentive-compatible for the senders.

\subsection{Theorem~\ref{thm:high-level-char}, general case}

The proof of the general case follows the same lines as the previous example. We show the generalization for the case of $k$-resilient incentive-compatibility, the proof for strong $k$-resilience is analogous, with the main differences highlighted in Section~\ref{sec:proof-strong}. 
In the example, note that we could argue that $m_d^*(\omega^1,\omega^1,\omega^2,\omega^2)$ should be greater than $m_d^*(\omega^1, \omega^1, \ldots, \omega^1)$  since, otherwise, senders 3 and 4 could increase their utility in state $\omega^1$ by reporting $\omega^2$ instead of $\omega^1$. More generally, suppose that in some state $\omega$ there exists a subset $C$ of at most $k$ senders such that all senders in $C$ prefer action $1$ to action $0$. Then, all $k$-resilient truthful mechanisms must satisfy that $m_d^*(\omega, \ldots, \omega) \ge m_d^*(\vec{m})$ for all inputs $\vec{m}$ such that $m_i = \omega$ for all $i \not \in C$. 

Following this intuition, we make the following definitions. Let $\Gamma = (S, A, \Omega, p, u)$ be an information aggregation game with $\Omega = \{\omega^1, \ldots, \omega^m\}$ and $|S| = n$. We say that a possible input $\vec{m} = (m_1, \ldots, m_n)$ for $m_d$ is $\omega$-\emph{pure} if $m_1 = m_2 = \ldots = m_n = \omega$ (i.e., if all $m_j$ are equal to $\omega$). We also say that an input is pure if it is $\omega$-pure for some $\omega$. Additionally, if $\omega \in \Omega$, we denote by $\vec{\omega}$ the $\omega$-pure input $(\omega, \ldots, \omega)$. Moreover, given two inputs $\vec{m} = (m_1, \ldots, m_n)$ and $\vec{m}' = (m'_1, \ldots, m'_n)$ for $m_d$, we say that $\vec{m} \prec_k \vec{m}'$ if the subset $C$ of senders such that their input differs in $\vec{m}$ and $\vec{m}'$ has size at most $k$, and such that 
\begin{itemize}
    \item[(a)] $\vec{m}$ is $\omega$-pure  for some $\omega$ and all senders in $C$ strictly prefer action $1$ to action $0$ in state $\omega$, or 
    \item[(b)] $\vec{m}'$ is $\omega$-pure  for some $\omega$ and all senders in $C$ strictly prefer action $0$ to action $1$ in state $\omega$. 
\end{itemize}
By construction we have the following property of $\prec_k$.

\begin{lemma}\label{lem:inequality}
A honest mechanism is $k$-resilient incentive-compatible for the senders if and only if $$\vec{m} \prec_k \vec{m}' \Longrightarrow m_d^*(\vec{m}) \le m_d^*(\vec{m}')$$
for all inputs $\vec{m}$ and $\vec{m}'$.
\end{lemma}

Note that Lemma~\ref{lem:inequality} completely characterizes the honest mechanisms that are $k$-resilient incentive-compatible for the senders. However, this lemma is of little use by itself since mechanisms have an exponential number of possible inputs. Let $\le_k$ be the partial order between pure states induced by $\prec_k$. More precisely, we say that two states $\omega$ and $\omega'$ satisfy $\omega \le_k \omega'$ if there exists a sequence of inputs $\vec{m}^1, \ldots, \vec{m}^t$ such that $$\vec{\omega} \prec_k \vec{m}^1 \prec_k \ldots \prec_k \vec{m}^t \prec_k \vec{\omega}'.$$
For instance, in the example at the beginning of this section, we would have that $\omega^1 \le_2 \omega^2$ since $(\omega^1,\omega^1,\omega^1,\omega^1) \prec_2 (\omega^1,\omega^1,\omega^2,\omega^2) \prec_2 (\omega^2,\omega^2,\omega^2,\omega^2)$. The following proposition shows that the $\le_k$ relations completely determine the outcomes implementable by honest mechanisms that are $k$-resilient incentive-compatible for the senders.

\begin{proposition}\label{prop:inequalities}
Let $\Gamma = (S, A, \Omega, p, u)$ be an information aggregation game. Then, an outcome $o$ of $\Gamma$ is implementable by an honest mechanism that is $k$-resilient incentive-compatible for the senders if and only if $$\omega \le_k \omega' \Longrightarrow o^*(\omega) \le o^*(\omega')$$
for all $\omega, \omega' \in \Omega$.
\end{proposition}

\begin{proof}
The fact that any honest mechanism that is $k$-resilient incentive-compatible for the senders implies $\omega \le_k \omega' \Longrightarrow o^*(\omega) \le o^*(\omega')$ follows directly from Lemma~\ref{lem:inequality}.

To show the converse, given $o$ satisfying $\omega \le_k \omega' \Longrightarrow o^*(\omega) \le o^*(\omega')$, define $m_d^*$ as follows.
If $\vec{m}$ is $\omega$-pure for some $\omega$, then $m_d^*(\vec{m}) := o^*(\omega)$. Otherwise, let $A^k_\prec(\vec{m})$ be the set of inputs $\vec{m}'$ such that $\vec{m} \prec_k \vec{m}'$ and $A^k_\succ(\vec{m})$ be the set of inputs $\vec{m}'$ such that $\vec{m}' \prec_k \vec{m}$. Then,
\begin{itemize}
    \item If $A^k_\prec(\vec{m}) = \emptyset$, then $m_d^*(\vec{m}) := 1$.
    \item Otherwise, if $A^k_\succ(\vec{m}) = \emptyset$, then $m_d^*(\vec{m}) := 0$.
    \item Otherwise, $$m_d^*(\vec{m}) := \frac{\min_{\vec{m}' \in A^k_\prec(\vec{m})}\{m_d^*(\vec{m}')\} + \max_{\vec{m}' \in A^k_\succ(\vec{m})}\{m_d^*(\vec{m}')\}}{2}.$$
\end{itemize}

Note that $m_d^*$ is well-defined since all elements in $A^k_\prec(\vec{m})$ and $A^k_\succ(\vec{m})$ are pure, which means that $m_d^*(\vec{m}')$ is already defined for all these elements. Moreover, the honest mechanism $M$ determined by $m_d^*$ implements $o$. It remains to show that $M$ is $k$-resilient incentive-compatible for the senders. By Lemma~\ref{lem:inequality} this reduces to show that $\vec{m} \prec_k \vec{m}' \Longrightarrow m_d^*(\vec{m}) \le m_d^*(\vec{m}')$
for all inputs $\vec{m}$ and $\vec{m}'$. To show this, take a pure input $\vec{\omega}$ and another input $\vec{m}$ such that $\vec{\omega} \prec_k \vec{m}$. 
If $\vec{m}$ is $\omega'$-pure, then $\vec{\omega} \prec_k \vec{m} \Longrightarrow \vec{\omega} \le_k \vec{\omega}'$ and thus $m_d^*(\vec{\omega}) \le m_d^*(\vec{\omega}')$.
If $\vec{m}$ is not pure and $A^k_\prec(\vec{m}) = \emptyset$ we have by construction that $m_d^*(\vec{m}) = 1$, which is greater than $m_d^*(\vec{\omega})$. Otherwise, for all $\omega'$ such that $\vec{\omega}' \in A^k_\prec(\vec{m})$, we have that $\omega \le_k \omega'$ and thus by assumption that $m_d^*(\vec{\omega}) \le m_d^*(\omega')$. Therefore,
$$\frac{\min_{\vec{m}' \in A^k_\prec(\vec{m})}\{m_d^*(\vec{m}')\}}{2} \ge \frac{m_d^*(\vec{\omega})}{2}$$
Moreover, we have that 
$$\frac{\max_{\vec{m}' \in A^k_\succ(\vec{m})}\{m_d^*(\vec{m}')\}}{2} \ge \frac{m_d^*(\vec{\omega})}{2}$$
since $\vec{\omega} \in A^k_\succ(\vec{m}')$. Hence 
$$m_d^*(\vec{m}) \ge m_d^*(\vec{\omega})$$ as desired. An analogous argument can be used for the case in which $\vec{m} \prec_k \vec{\omega}$.
\end{proof}

It remains to show that the partial order between the states in $\Omega$ defined by $\le_k$ can be computed with a polynomial algorithm. To do this, note that, by definition, any chain $$\vec{\omega} \prec_k \vec{m}^1 \prec_k \ldots \prec_k \vec{m}^t \prec_k \vec{\omega}'$$
between two pure inputs $\vec{\omega}$ and $\vec{\omega}'$ must satisfy that either $\vec{m}^1$ or $\vec{m}^2$ are also pure. This implies the following lemma:

\begin{lemma}\label{lem:simplification}
Let $\Gamma = (S, A, \Omega, p, u)$ be an information aggregation game with $\Omega = \{\omega^1, \ldots, \omega^m\}$. Let $E$ a system of equations over $x_1, \ldots, x_m$ such that equation $x_i \le x_j$ appears in $E$ if and only if $\vec{\omega}^i \prec_k \vec{\omega}_j$ or if there exists an input $\vec{m}$ such that $\vec{\omega}^i \prec_k \vec{m} \prec_k \vec{\omega}^j$. Then, $y_1, \ldots, y_m$ is a solution of $E$ if and only if $$\omega^i \le_k \omega^j \Longrightarrow y_i \le y_j$$ for all $i, j \in [m]$.
\end{lemma}

Intuitively, Lemma~\ref{lem:simplification} says that the inequalities obtained from chains of length 2 or 3 \emph{span} the partial order over $\Omega$ defined by $\le_k$, and thus that we can take the system of equations $E$ of Theorem~\ref{thm:high-level-char} to be the one in the lemma above. Therefore, given two states $\omega$ and $\omega'$, it only remains to show that we can check in polynomial time if $\vec{\omega} \prec_k \vec{\omega}'$ or if there exists a state $\vec{m}$ such that $\vec{\omega} \prec_k \vec{m} \prec_k \vec{\omega}'$. Checking if $\vec{\omega} \prec_k \vec{\omega}'$ is equivalent to checking if $k = n$ and either all senders prefer $1$ in $\omega$ or all senders prefer $0$ in $\omega'$. Finding an input $\vec{m}$ such that $\vec{\omega} \prec_k \vec{m} \prec_k \vec{\omega}'$ reduces to finding an input $\vec{m}$ such that 
\begin{itemize}
    \item [(a)] the set $C_\omega$ of senders such that their message is not $\omega$ in $\vec{m}$ has size at most $k$, and all senders in $C_\omega$ strictly prefer $1$ to $0$ in $\omega$.
    \item [(b)] the set $C_{\omega'}$ of senders such that their message is not $\omega'$ in $\vec{m}$ has size at most $k$, and all of them strictly prefer $0$ to $1$ in $\omega'$.
\end{itemize}

The high level idea of the algorithm is that, if $\vec{m}$ satisfies the above properties, all senders $i$ that prefer $0$ to $1$ in $\omega$ must satisfy that $m_i = \omega$ (otherwise, it breaks property (a)), and all senders $i$ that prefer $1$ to $0$ in $\omega'$ must satisfy that $m_i = \omega'$ (otherwise, it breaks property (b)). If there is a sender $i$ that prefers $0$ to $1$ in $\omega$ and $1$ to $0$ in $\omega'$ then such an input $\vec{m}$ does not exist, and if there is a sender $i$ that strictly prefers $1$ to $0$ in $\omega$ and $0$ to $1$ in $\omega'$, then $m_i$ has no constraints. The only remaining restriction is that there can only be at most $k$ values different than $\omega$ and at most $k$ values different than $\omega'$ (note that this implies that if $2k < n$ such an input does not exist). The algorithm goes as follows:

\begin{enumerate}
    \item Split the set of senders into four subsets $X_{0,1}^{0,1}, X_{0,1}^{1,0}, X_{1, 0}^{0,1}, X_{1, 0}^{1,0}$, in which $X_{i,j}^{i', j'}$ is the set of senders that prefer $i$ to $j$ in $\omega$ (resp., strictly prefer if $i = 1$) and prefer $i'$ to $j'$ in $\omega'$ (resp., strictly prefer if $i' = 0$).
    \item If $X_{0,1}^{1,0} \not = \emptyset$ or $2k < n$, there is no solution.
    \item If $|X_{0,1}^{0,1}| > k$ or $|X_{1,0}^{1,0}| > k$, there is no solution.
    \item Otherwise, set $m_i = \omega$ for all $i \in X_{0,1}^{0,1}$, $m_i = \omega'$ for all $i \in X_{1,0}^{1,0}$. Then, set $k - |X_{0,1}^{0,1}|$ of the messages from $X_{1,0}^{0,1}$ to $\omega$ and the rest to $\omega'$. Return $\vec{m}$.
\end{enumerate}

\textbf{Proof of Correctness:} Because of the previous discussion, if $X_{0,1}^{1,0} \not = \emptyset$ or $2k < n$, there is no solution. If $|X_{0,1}^{0,1}| \ge k$ then, any input $\vec{m}$ that satisfies $\vec{\omega} \prec_k \vec{m} \prec_k \vec{\omega}'$ would require to have at least $|X_{0,1}^{0,1}|$ components equal to $\omega$, which would break property (b). An analogous argument can be used when $|X_{1,0}^{1,0}| > k$. If none of these conditions hold, then we set all messages from $X_{0,1}^{0,1}$ to $\omega$, all messages from $X_{1,0}^{1,0}$ to $\omega'$, and we split the messages sent by senders in $X_{1,0}^{0,1}$ between $\omega$ and $\omega'$ in such a way that no value appears more than $k$ times. The resulting input satisfies properties (a) and (b). 

\subsection{Theorem~\ref{thm:high-level-char}, strong $k$-resilience}\label{sec:proof-strong}

The proof of Theorem~\ref{thm:high-level-char} for strong $k$-resilience is analogous to the one of $k$-resilience in the previous section. The main difference is the definition of $\prec_k$. In this case we say that two inputs $\vec{m}$ and $\vec{m}'$ satisfy $\vec{m} \prec_k^s \vec{m}'$ if and only if the subset $C$ of senders such that their input differs in $\vec{m}$ and $\vec{m}'$ has size at most $k$, and such that 
\begin{itemize}
    \item[(a)] $\vec{m}$ is $\omega$-pure  for some $\omega$ and at least one sender in $C$ strictly prefers action $1$ to action $0$ in state $\omega$, or 
    \item[(b)] $\vec{m}'$ is $\omega$-pure  for some $\omega$ and at least one sender in $C$ strictly prefers action $0$ to action $1$ in state $\omega$. 
\end{itemize}

We have that $\vec{\omega} \prec_k^s \vec{\omega}'$ if and only if $k = n$ and at least one sender in $\omega$ prefers action $1$ to action $0$, or at least one sender in $\omega'$ prefers action $0$ to action $1$. Given $\omega$ and $\omega'$, finding if there exists $\vec{m}$ such that $\vec{\omega} \prec_k^s \vec{m} \prec_k^s \vec{\omega}'$ can be reduced to finding if there exists a partition of the set of senders $S$ into two sets $S_\omega$ and $S_{\omega'}$ such that $|S_\omega| \le k$ and $|S_{\omega'}| \le k$, and such that at least one sender of $S_\omega$ prefers action $0$ to $1$ in $\omega'$  and at least one sender of $S_{\omega'}$ prefers $1$ to $0$ in $\omega$. This can easily be done in polynomial time. 

For future reference, we define $\le_k^s$ in the same way as $\le_k$ except that we use $\prec_k^s$ instead of $\prec_k$. 

\section{Proof of Theorem~\ref{thm:construction}}\label{sec:proof-construction}

Most of the tools used to prove Theorem~\ref{thm:construction} have already appeared in the proof of Theorem~\ref{thm:high-level-char}. We prove the theorem for $k$-resilience, the case of strong $k$-resilience is analogous. Given a game $\Gamma$ and an outcome $o$ for $\Gamma$, we set $m_d^*(\vec{\omega}) := o^*(\omega)$ for each $\omega \in \Omega$. For every other input $\vec{m}$, we define $m_d^*(\vec{m})$ in the same way as in the proof of Proposition~\ref{prop:inequalities}. As shown in the proof of Theorem~\ref{thm:high-level-char}, checking if $\vec{m} \prec_k \vec{m}'$ can be performed in polynomial time. Thus, $m_d^*(\vec{m})$ can also be computed in polynomial time.

\section{Extended Model and Generalization of Main Results}\label{sec:extended}

An \emph{extended information aggregation game} is defined in the same way as a standard information aggregation game (see Section~\ref{sec:model}) except that each sender starts the game with a private signal $x_i$ (as opposed to all senders starting the game with the same input $\omega$), and the utility function $u$ takes as input the signals from each sender instead of just $\omega$. More precisely, in an extended information aggregation game $\Gamma = (S, A, X, p, u)$ there is a set of senders $S = \{1,2,3,\ldots, n\}$, a receiver $r$, a mediator $d$, a set of actions $A$, a set $X = X_1 \times X_2 \times \ldots \times X_n$ of signals, a probability distribution $p$ over $X$, and a utility function $u : (S \cup \{r\}) \times X \times A \longrightarrow \mathbb{R}$. Each game instance proceeds exactly the same way as in a standard information aggregation game except that, in phase 1, a signal profile $(x_1, \ldots, x_n) \in X$ is sampled following distribution $p$, and each signal $x_i$ is disclosed only to sender $i$. In this context, an outcome $o$ for $\Gamma$ is just a function from signal profiles $\vec{x} \in X$ to distributions over $A$, and mechanisms for $\Gamma$ are determined by functions $m_d^*$ from $X$ to $[0,1]$.

Our aim is to generalize the results from Section~\ref{sec:results} to the extended model. However, the
main problem is that, for a fixed signal profile, the preferences of the agents may depend on their coalition. For instance, consider a game $\Gamma$ for five players with uniformly distributed binary signals and binary actions such that the utility of each sender is $1$ if the action that the receiver plays is equal to the majority of the signals, and their utility is $0$ otherwise. Suppose that senders have signals $(0,0,0,1,1)$. It is easy to check that if players $1$, $2$ and $3$ collude, player $1$ would prefer action $0$ to action $1$. However, if players $1$, $4$ and $5$ collude, player $1$ would prefer action $1$ since in this case it is more likely that the majority of the signals are $1$.

We can avoid the issue above by assuming that the game is \emph{$k$-separable}, which is that, for all signal profiles $\vec{x}$ and all senders $i$, there exists an action $a$ such that the preference of sender $i$ inside any coalition $K$ of size at most $k$ is $a$. Intuitively, an extended information aggregation game is \emph{$k$-separable} if the preferences of the senders do not depend on the coalition they are in. With this, we can provide algorithms for the characterization and implementation of $k$-resilient truthful implementable outcomes that are efficient relative to the  size of the description of the game $\Gamma$.

\begin{theorem}\label{thm:high-level-desc-extended}
Let $\Gamma = (S, A, X, p, u)$ be a $k$-separable extended information aggregation game such that the support of signal profiles in distribution $p$ is $\{(\vec{x})_1, \ldots, (\vec{x})_m\}$. Then, there exists a system $E$ of $O(m^2)$ equations over variables $x_1, \ldots, x_m$, such that each equation of $E$ is of the form $x_i \le x_j$ for some $i,j \in [m]$, and such that an outcome $o$ of $\Gamma$ is implementable by a $k$-resilient truthful mechanism (resp., strong $k$-resilient truthful mechanism) if and only if 

\begin{itemize}
\item[(a)] $x_1 = o^*((\vec{x})_1), \ldots, x_m = o^*((\vec{x})_m)$ is a solution of $E$.
\item[(b)] $E_r\left(o\right) \ge U_a$ for all $a \in A$.
\end{itemize}

Moreover, the equations of $E$ can be computed in polynomial time over $m$ and the number of senders $n$.
\end{theorem}

Note that Theorem~\ref{thm:high-level-desc-extended} states that $E$ can be computed in polynomial time over the size of the support of signal profiles as opposed to $|X|$, which may be way larger. There is also a generalization of Theorem~\ref{thm:construction} in the extended model.

\begin{theorem}\label{thm:construction-extended}
There exists an algorithm $A$ that receives as input the description of a $k$-separable extended information aggregation game  $\Gamma = (S, A, \Omega, p, u)$, an outcome $o$ for $\Gamma$ implementable by a $k$-resilient  mechanism (resp., strong $k$-resilient mechanism), and a message input $\vec{m}$ for the mediator, and $A$ outputs a value $q \in [0,1]$ such that the function $m_d^*$ defined by $m_d^*(\vec{m}) := A(\Gamma, o, \vec{m})$ determines a $k$-resilient truthful mechanism (resp., strong $k$-resilient truthful mechanism) for $\Gamma$ that implements $o$. Moreover, $A$ runs in polynomial time over the size $m$ of the support of signal profiles and $|S|$.
\end{theorem}

The proofs of Theorems~\ref{thm:high-level-desc-extended} and \ref{thm:construction-extended} are analogous to the ones of Theorems~\ref{thm:high-level-char} and \ref{thm:construction} with the following difference. Given two inputs $\vec{m}$ and $\vec{m}'$, we say that $\vec{m} \prec_k \vec{m}'$ if the subset $C$ of senders such that their input differs in $\vec{m}$ and $\vec{m}'$ has size at most $k$, and such that 
\begin{itemize}
    \item[(a)] $\vec{m}$ is in the support of $p$ and all senders in $C$ strictly prefer action $1$ to action $0$ given signal profile $\vec{m}$, or 
    \item[(b)] $\vec{m}'$ is is in the support of $p$ and all senders in $C$ strictly prefer action $0$ to action $1$ given signal profile $\vec{m}'$. 
\end{itemize}
 
 Intuitively, we replace the notion of \emph{pure input} by the condition that the input is in the support of $p$. Note that the assumption of $k$-separability is crucial for this definition, since otherwise the preferences of the players may not be uniquely determined by the signal profile. With this definition, we can construct analogous statements for Lemmas~\ref{lem:inequality}, \ref{lem:simplification} and Proposition~\ref{prop:inequalities}, and proceed identically as in the proofs of Theorems~\ref{thm:high-level-char} and \ref{thm:construction}.

\section{Conclusion}

We provided an efficient characterization of all outcomes implementable by $k$-resilient and strong $k$-resilient truthful mechanisms in information aggregation games. We also gave an efficient construction of the $k$-resilient or strong $k$-resilient mechanism that implements a given implementable outcome. These techniques generalize to the extended model where senders may receive different signals, as long as the senders' preferences are not influenced by their coalition ($k$-separability). It is still an open problem to find if the techniques used in this paper generalize to other notions of coalition resilience as, for instance, the notion in which the sum of utilities of the members of a coalition cannot increase when defecting, or if we can get efficient algorithms in the extended model without the $k$-separability assumption. It is also an open problem to find if we can get similar results in partially synchronous or asynchronous systems in which the messages of the senders are delayed arbitrarily.

\bibliographystyle{eptcs}
\bibliography{bibfile}

\end{document}